\documentclass[aps,pra,twocolumn,showpacs,groupedaddress]{revtex4}

\usepackage{graphicx}
\usepackage{color}
\usepackage{amsthm}
\usepackage[toc,page,title,titletoc,header]{appendix}
\usepackage{CJK}
\theoremstyle{plain}
\newtheorem{thm}{Theorem}

\theoremstyle{definition}

\theoremstyle{remark}
\newtheorem*{rem}{Remark}

\newcommand{\Tr}{\text{Tr}}
\newcommand{\ket}[1]{|#1\rangle}
\newcommand{\bra}[1]{\langle#1|}
\begin{document}
\title{Time-dependent Decoherence-Free Subspace}

\author{S. L. Wu,  L. C. Wang, X. X. Yi}

\email{yixx@dlut.edu.cn}

\affiliation{School of Physics and Optoelectronic Technology,\\
Dalian University of Technology, Dalian 116024 China}

\date{\today}
\begin{abstract}
With time-dependent Lindblad operators, an open system may have a
time-dependent decoherence-free subspace (t-DFS). In this paper, we
define the t-DFS and  present a  necessary and sufficient condition
for the t-DFS.  Two examples are presented to illustrate the t-DFS,
which show that this t-DFS is not trivial, when the dimension of the
t-DFS varies.
\end{abstract}

\pacs{03.65.Yz, 03.67.-a, 03.65.Ta } \maketitle


\section{Introduction}

In recent years, many  proposals are presented  to protect quantum
systems against decoherence. Except the method to weaken the
coupling between the system and its surroundings, these proposals
include the dynamical decoupling \cite{dd1,dd2,dd3}, quantum
error-correcting codes \cite{cc1,cc2}, the scheme based on the
decoherence-free subspaces (DFSs) or noiseless subsystems
\cite{dfs1,dfs2,dfs3,dfs4}, and the scheme based on  the quantum
reservoir engineering \cite{qe1,qe2}.

The  DFS is defined as a subspace within which the system undergoes
an unitary evolution \cite{dfs4}. It was experimentally realized in
a variety of systems \cite{dfse1,dfse2,dfse3} and has drawn much
attention because its potential applications in quantum information
processing \cite{dfsa1,dfsa2}. Generally speaking, there are two
ways to define the DFS. {The first is given in \cite{dfs2}, where
the DFS includes all states that each state $\rho$ satisfies
$\mathcal{L}(\rho)=0$, where $\mathcal{L}(...)$ represents the
Lindblad superoperator (we call it as the first definition of DFS).}
The second definition was given in \cite{dfs4}, which is formulated
as follows. Let the time evolution of an open system with Hilbert
space $\mathcal H_S$ be governed by the Markovian master equation. A
decoherence-free subspace $\mathcal H_{\text{DFS}}$ is defined as a
subspace of $\mathcal H_S$ such that all states $\rho(t)$ in DFS
fulfill $
\partial_t
\Tr[\rho^2(t)]=0,\,\,$ for
$\forall\,t\geq0,\,\,\text{with}\,\Tr[\rho^2(0)]=1. $ By this
definition, it was proved  that the subspace $ \mathcal
H_{\text{DFS}}=\text{Span}\{\ket{\Phi_1},\ket{\Phi_2},\cdot\cdot\cdot,\ket{\Phi_M}\}
$ is a  DFS if and only if each basis of $\mathcal H_{\text{DFS}}$
satisfies $ F_\alpha\ket{\Phi_j}=c_\alpha\ket{\Phi_j}, j=1,...,M;
\alpha=1,...,K, $ and $\mathcal H_{\text{DFS}}$ is invariant under $
H_{\text{eff}}= H+\frac{i}{2}\sum_\alpha\left(c^*_\alpha
F_\alpha-c_\alpha F^\dag_\alpha\right), \label{heff} $ where $H$ is
the Hamiltonian of the open quantum system.

Notice that the basis of the aforementioned DFS is time-independent,
the DFS is then time-independent. For time-dependent Lindblad
operators, however, a time-independent DFS may not exist, then a
natural question arises:{What is the DFS for open systems with
time-dependent Lindblad operators? If the DFS is time-dependent,
what is the condition for the system to remain in this subspace?}

In this paper, we will give a detail analysis for a time-dependent
DFS (t-DFS). The analysis is given based on the Lindblad master
equation with time-dependent Lindblad operators. We shall adopt the
second definition for the  t-DFS  and develop a theorem to give a
necessary and sufficient condition for the t-DFS.

This paper is organized as follows. In Sec.\ref{SNC}, we shall give
the definition for the t-DFS and derive a  necessary and sufficient
condition for it. In Sec.\ref{example}, we present two  examples to
illustrate the t-DFS, showing that the necessary and sufficient
condition can be satisfied by manipulating the Hamiltonian. Finally
we conclude by summarizing our results in Sec.\ref{conclusion}.

\section{necessary and sufficient condition for time-dependent DFS}\label{SNC}

In this section, we will present a necessary and sufficient
condition for the t-DFS and show that this condition can be
satisfied by manipulating the Hamiltonian of the open system.

Consider a system $S$ with $N$-dimensional Hilbert space $\mathcal
H_S$ coupling with an environment.  The time evolution of the open
system is assumed to be governed  by,
\begin{eqnarray}
\dot{\rho}(t)&=&-i[H(t),\rho(t)]+\mathcal{L}(t)\rho(t),\nonumber\\
\mathcal{L}(t)\rho(t)&=&\sum_\alpha \left[F_\alpha(t) \rho(t)
F_\alpha^\dag(t)-\frac{1}{2}\{F_\alpha^\dag(t) F_\alpha(t),
\rho(t)\}\right],\nonumber\\ \label{tmq}
\end{eqnarray}%
where $F_\alpha(t)$ $(\alpha=1,2,3...,K)$  are time-dependent
Lindblad operators {and $\mathcal{L}(t)$ describes the Lindblad
superoperator of the open quantum system.}

\begin{thm}
Let the time evolution of an open quantum system in a
finite-dimensional Hilbert space be governed by Eq.(\ref{tmq}) with
time-dependent Hamiltonian $ H(t)$ and time-dependent Lindblad
operators $F_\alpha(t)$. The subspace
\begin{eqnarray}
\mathcal
H_{\text{DFS}}(t)=\text{Span}\{\ket{\Phi_1(t)},\ket{\Phi_2(t)},
\cdot\cdot\cdot,\ket{\Phi_M(t)}\}
\end{eqnarray}%
is a t-DFS if and only if each basis vector of $\mathcal
H_{\text{DFS}}(t)$ satisfies
\begin{eqnarray}
F_\alpha(t)\ket{\Phi_j(t)}=c_\alpha(t)\ket{\Phi_j(t)}, j=1,...,M;
\alpha=1,...,K, \nonumber\\ \label{sandn}
\end{eqnarray}%
and $\mathcal H_{\text{DFS}}(t)$ is invariant under
\begin{eqnarray}
H_{\text{eff}}(t)&=&G(t)+
H(t)\nonumber\\&&+\frac{i}{2}\sum_\alpha\left(c^*_\alpha(t)
F_\alpha(t)-c_\alpha(t) F^\dag_\alpha(t)\right). \label{heff}
\end{eqnarray}%
Here $G(t)=iU^\dag(t)\dot U(t)$  and $U(t)$ is an unitary operator
\begin{eqnarray}
 U(t)=\sum_{j=1}^M
\ket{\Phi_j(0)}\bra{\Phi_j(t)}+\sum_{n=1}^{N-M}
\ket{\Phi_n^\bot(0)}\bra{\Phi_n^\bot(t)}.\label{unitary}
\end{eqnarray}%
\end{thm}

\begin{proof}
Firstly, notice that the effect of $U(t)$ is to map a set of
time-dependent bases of $\mathcal H_S$ into a  time-independent one.
Transforming the density matrix $\rho(t)$ into a rotating frame,
i.e., $\bar \rho(t)=U(t)\rho(t)U^\dag(t)$, we write the master
equation Eq.(\ref{tmq}) as,
\begin{eqnarray}
\frac{\partial\bar\rho(t)}{\partial t}&=&-i\left[\bar{
H}(t),\bar\rho(t)\right]\nonumber\\&&-i \bar
G(t)\bar\rho(t)+i\bar\rho(t) \bar G(t)+\bar{\mathcal
L}(t)\bar\rho(t),\label{meq2}
\end{eqnarray}%
where, $$\bar{H}(t)=U(t) H(t) U^\dag(t),$$ $$\bar{\mathcal
L}(t)\bar\rho(t)=\sum_\alpha(\bar F_\alpha(t)\bar\rho(t)\bar
F_{\alpha}^\dag(t)-1/2\{\bar F_{\alpha}^\dag(t)\bar
F_\alpha(t),\bar\rho(t)\})$$ with $\bar{F}_\alpha(t)= U(t)
F_\alpha(t) U^\dag(t)$. {Clearly, $\bar
G(t)=U(t)G(t)U^\dag(t)=i\dot{ U}(t){ U}^\dag(t)=-i
U(t)\dot{U}^\dag(t)=\bar G^\dag(t)$, this indicates that $\bar G(t)$
is a Hermitian operator.} By defining new Lindlbad operator as
$\tilde F_\alpha(t)=\bar F_\alpha(t)-c_\alpha(t)$, the decoherence
terms  in Eq.(\ref{meq2}) can be rewritten as
\begin{eqnarray}
\bar{\mathcal L}(t)\bar\rho(t)&=&\tilde{\mathcal
L}(t)\bar\rho(t)\nonumber\\
&-&i\left[\frac{i}{2}\sum_\alpha(c^*_\alpha(t)\bar
F_\alpha(t)-c_\alpha(t)\bar
F^\dag_\alpha(t)),\bar\rho(t)\right],\nonumber\\ \label{supero}
\end{eqnarray}
where
$$\tilde{\mathcal L}(t)\bar\rho(t)=\sum_\alpha(\tilde
F_\alpha(t)\bar\rho(t)\tilde F_{\alpha}^\dag(t)-1/2\{\tilde
F_{\alpha}^\dag(t)\tilde F_\alpha(t),\bar\rho(t)\}).$$

Note that the requirement of invariance of the subspace $\mathcal
H_{\text{DFS}}(t)$ under the operator $ H_{\text{eff}}(t)$ in
Theorem 1 implies $\bra{\Phi^\bot_n(0)}\bar
H_{\text{eff}}(t)\ket{\Phi_j(0)}$$=\bra{ \Phi^\bot_n(t)}
H_{\text{eff}}(t)\ket{\Phi_j(t)}=0$ for $\forall\, n,\,j$. In the
rotating frame, the subspace $\bar \mathcal H_{\text{DFS}}$ spanned
by  $\{\ket{\Phi_j(0)}\}$ must be invariant under the operator
\begin{eqnarray}
\bar{ H}_{\text{eff}}(t)&=&\bar{{H}}(t)+\bar
G(t)\nonumber\\&&+\frac{i}{2}\sum_\alpha\left(c^*_\alpha(t)\bar
F_\alpha(t)-c_\alpha(t)\bar F^\dag_\alpha(t)\right).
\end{eqnarray}%

Now we prove that the condition is sufficient. Any state
$\ket{\varphi(t)}\in\mathcal H_{\text{DFS}}(t)$ can be expanded  by
$\{\ket{\Phi_j(t)}\}$,
\begin{eqnarray}
\ket{\varphi(t)}=\sum_{j=1}^M a_j(t)\ket{\Phi_j(t)}.
\end{eqnarray}%
By defining $\ket{\bar\varphi(t)}=\text U(t)\ket{\varphi(t)},$
$\bar\rho(t)$ and $\bar{\mathcal L}(t)\bar\rho(t)$ will be written
as,
\begin{eqnarray}
&\bar\rho(t)= U(t)\ket{\varphi(t)}\bra{\varphi(t)}
U^\dag(t),\nonumber\\
&\bar{\mathcal
L}(t)\bar\rho(t)=-i\left[\frac{i}{2}\sum_\alpha(c^*_\alpha(t)\bar
F_\alpha(t)-c_\alpha(t)\bar
F^\dag_\alpha(t)),\bar\rho(t)\right],\nonumber\\
\end{eqnarray}%
where $\tilde
F_\alpha(t)\ket{\bar\varphi(t)}=0\times\ket{\bar\varphi(t)}=0$ and
$\tilde{\mathcal L}(t) \bar\rho(t)=0$ have been used. Hence the
evolution of $\bar\rho(t)$ is governed by
\begin{eqnarray}
\frac{\partial{\bar\rho}(t)}{\partial t}=-i[\bar{\text
H}_{\text{eff}}(t),\bar\rho(t)],\label{UT}
\end{eqnarray}%
so
\begin{eqnarray}
\frac{\partial}{\partial
t}{\Tr[\rho^2(t)]}&&=\frac{\partial}{\partial
t}{\Tr[\bar\rho^2(t)]}=2\Tr[\bar\rho(t)\dot{\bar\rho}(t)]\nonumber\\
&&=-2i\Tr\{[\bar{\text
H}_{\text{eff}}(t),\bar\rho(t)]\,\bar\rho(t)\}=0.
\end{eqnarray}%
Thus the conditions Eqs.(\ref{sandn}) and (\ref{heff}) are
sufficient for $\frac{\partial}{\partial t}{\Tr[\rho^2(t)]}=0$.

In order to prove that the conditions are necessary, we assume that
the set of basis $\{\ket{\Phi_j(t)}\}$ spans a t-DFS. At a fixed
instant $t_0$, the quantum state $\ket{\varphi(t_0)}$ embeds in
t-DFS
\begin{eqnarray}
&\ket{\varphi(t_0)}=\sum_{j=1}^M a_j(t_0) \ket{\Phi_j(t_0)},
\nonumber\\
&\rho(t_0)= \ket{\varphi(t_0)}\bra{\varphi(t_0)},
\end{eqnarray}%
which implies
\begin{eqnarray}
0&=&\frac{\partial}{\partial
t}{\Tr[\rho^2(t)]}|_{t=t_0}=2\Tr[\rho(t_0)\,{\mathcal
L}(t_0)\rho(t_0)]\nonumber\\
&&=2\bra{\varphi(t_0)} {\mathcal L}(t_0)\rho(t_0)
\ket{\varphi(t_0)}.
\end{eqnarray}%
Then, by using Eq.(\ref{tmq}), we obtain
\begin{eqnarray}
0&=&\bra{\varphi(t_0)}{\mathcal
L}(t_0)\rho(t_0) \ket{\varphi(t_0)}\nonumber\\
&=&\sum_\alpha\gamma_\alpha\left[\bra{\varphi(t_0)}F_\alpha(t)\ket{\varphi(t_0)}\bra{\varphi(t_0)}
F_\alpha^\dag(t)\ket{\varphi(t_0)}\right.\nonumber\\
&&\left.-\bra{\varphi(t_0)}F_\alpha^\dag(t)F_\alpha(t)\ket{\varphi(t_0)}\right].\label{ff}
\end{eqnarray}%
Without loss of generality, we set
$F_\alpha(t_0)\ket{\varphi(t_0)}=c_\alpha(t_0)\ket{\varphi(t_0)}+\ket{\varphi^\bot(t_0)}$
with $\ket{\varphi^\bot(t_0)}$ being a  (non-normalized) state
orthogonal to state $\ket{\varphi(t_0)}$. Submitting this into
Eq.(\ref{ff}), we have
\begin{eqnarray}
\sum_\alpha\gamma_\alpha\langle\varphi^\bot(t_0)\ket{\varphi^\bot(t_0)}=0.
\end{eqnarray}%
Since $\gamma_\alpha>0$ for $\forall\alpha$, we have
$\langle\varphi^\bot(t_0)\ket{\varphi^\bot(t_0)}=0$. It
straightforwardly follows
$F_\alpha(t_0)\ket{\varphi(t_0)}=c_\alpha(t_0)\ket{\varphi(t_0)}$.

Next,  we prove by contradiction that the basis vectors
$\{\ket{\Phi_j(t_0)}\}$ are the eigenstates of the Lindblad
operators $F_\alpha(t_0)$ with the same eigenvalues $c_\alpha(t_0)$.
Suppose that two arbitrary eigenvectors of the Lindblad operator
$\ket{\Phi_k(t)}$ and $\ket{\Phi_{k'}(t)}$ have different
eigenvalues, i.e.,
\begin{eqnarray}
&F_\alpha(t)\ket{\Phi_k(t)}=c_{\alpha,k}(t)\ket{\Phi_k(t)},\nonumber\\
&F_\alpha(t)\ket{\Phi_{k'}(t)}=c_{\alpha,k'}(t)\ket{\Phi_{k'}(t)}.
\end{eqnarray}%
The state $\ket{\phi(t)}=(\ket{\Phi_k(t)}+\ket{\Phi_{k'}(t)})/\sqrt
2$ must be not an eigenstate of $F_\alpha(t)$. However, since
$\ket{\phi(t)}$ is a state lying in $\mathcal H_{\text{DFS}}(t)$,
the state should fulfill $\langle \mathcal L
[\ket{\phi(t)}\bra{\phi(t)}]\rangle=0$. Hence the eigenvalues must
be equal for all basis vectors.

On the other hand, according to Eq.(\ref{supero}), for $\bar\rho(t)=
U(t) \ket{\varphi(t)}\bra{\varphi(t)}  U^\dag(t)$ with $\bar
F_\alpha(t)\ket{\bar \varphi(t)}=c_\alpha(t)\ket{\bar \varphi(t)}$,
we have
\begin{eqnarray}
\frac{\partial{\bar\rho}(t)}{\partial t}&=&-i \bar
G(t)\bar\rho(t)+i\bar\rho(t) \bar G(t)-i\left[\bar{\text
H}(t),\bar\rho(t)\right]\nonumber\\
&&+\bar{\mathcal
L}(t)\bar\rho(t)\nonumber\\
&=&-i[\bar{ H}_{\text{eff}}(t),\bar\rho(t)].
\end{eqnarray}
and
\begin{widetext}
\begin{eqnarray}
\bar\rho(t)=\mathcal T\exp\left[-i\int_o^t \bar{\text
H}_{\text{eff}}(\tau)\text d \tau\right]\bar\rho(0)\,\mathcal
T\exp\left[i\int_o^t \bar{ H}_{\text{eff}}(\tau)\text d \tau\right],
\end{eqnarray}
\end{widetext}
where $\mathcal T$ is the time-order operator. If $\ket{\bar
\psi(t)}=\bar{H}_{\text{eff}}(t)\ket{\bar \varphi(t)}$ is still  in
subspace $\bar \mathcal H_{\text{DFS}}$, $\ket{\bar\varphi(t)}$ will
always be in $\bar \mathcal H_{\text{DFS}}$. Thus,
$\ket{\varphi(t)}$ will evolve unitarily in $\mathcal
H_{\text{DFS}}(t)$, if $ H_{\text{eff}}(t)\ket{\varphi(t)}$ is still
a superposition of the basis vectors $\{\ket{\Phi_j(t)}\}$ of the
t-DFS. Hence the conditions are necessary for
$\frac{\partial}{\partial t}{\Tr[\rho^2(t)]}=0$.
\end{proof}
\begin{rem}
Examining  the unitary transformation Eq.(\ref{unitary}), one may
suspect that the t-DFS  is the same as (time-independent) DFS, in
the sense that a state in the t-DFS $\mathcal H_{\text{DFS}}(t)$ can
be obtained from the time-independent DFS $\bar\mathcal
H_{\text{DFS}}$ by the unitary transformation $U(t)$. This is not
true, because for a t-DFS, its dimension  can change with time. When
the dimension of the t-DFS changes, we can not transfer the
time-independent DFS into a t-DFS by an unitary transformation. In
Sec.\ref{example}, we will present an example to illustrate this
$remark$.
\end{rem}

We now show how to realize a t-DFS by the Theorem 1. Suppose that
there is a set of degenerated eigenstates $\{\ket{\Phi_j(t)}\}$ of
all Lindblad operators $F_\alpha(t)$. As stated above, if $\mathcal
H_{\text{DFS}}(t)$ is a t-DFS, it must be invariant under $
H_{\text{eff}}(t)$, i.e.,
\begin{eqnarray}
\bra{\Phi^\bot_n(t)}\text
H_{\text{eff}}(t)\ket{\Phi_j(t)}=0,\forall\, n,\,j.\label{x1}
\end{eqnarray}
Substituting Eq.(\ref{heff}) into Eq.(\ref{x1}) and considering
$F_\alpha(t)\ket{\Phi_j(t)}=c_\alpha(t)\ket{\Phi_j(t)}$, we obtain
\begin{eqnarray}
-\bra{\Phi_n^\bot(t)}\dot{\Phi}_j(t)\rangle
-i\bra{\Phi^\bot_n(t)}{H}(t)\ket{\Phi_j(t)}&&\nonumber\\
-\frac{1}{2}\sum_\alpha\gamma_\alpha
c_\alpha(t)\bra{\Phi_n^\bot(t)}F^\dag_\alpha(t)\ket{\Phi_j(t)}&=&0.
\end{eqnarray}
Once the evolution of the Lindblad operators $F_\alpha(t)$ are
fixed, the basis vectors of $\mathcal H_{\text{DFS}}(t)$ are
specified. The task to achieve an unitary evolution in the t-DFS
relies entirely on the  Hamiltonian of the open system,
\begin{eqnarray}
&&\bra{\Phi_k(t)}{H}(t)\ket{\Phi^\bot_n(t)} =-i\bra{\dot{\Phi}_k(t)}
\Phi_n^\bot(t)\rangle\nonumber\\
&&-\frac{i}{2}\sum_\alpha\gamma_\alpha
c^{\ast}_\alpha(t)\bra{\Phi_k(t)}F_\alpha(t)\ket{\Phi_n^\bot(t)}.\label{relation}
\end{eqnarray}
In the next section, we will present two examples to show that the
quantum state of the open system can be restricted to evolve
unitarily in the t-DFS by manipulating the Hamiltonian according to
Eq.(\ref{relation}).

\section{examples}\label{example}

In this section, we will present two examples. In the first example,
we show that the sufficient and necessary condition for the t-DFS
provides us a way to keep the open system in the t-DFS by
manipulating the system Hamiltonian. In contrast to the scheme  in
Ref.\cite{dfe1}, where the adiabatic condition is required to
maintain the system with high fidelity in the t-DFS, here the
adiabatic condition is removed. Here we should mention that some
effort has been done on such issue, in which the time-dependent
Lindblad operators evolves unitarily\cite{tdf}. However, for the
t-DFS, such a limit on the evolution of Lindblad operators is not
necessary, therefore t-DFS we proposed is general. From the other
aspect, by the first example, we explain why the adiabatic condition
is required in Ref.\cite{dfe1}--the Hamiltonian in that paper does
not satisfy the condition of t-DFS. It seems that the t-DFS can be
found by an unitary transformation, this is not the case when the
dimension of the t-DFS changes, this point will be shown in  the
second example.

\subsection{Driven $\Xi$-type three-level atom coupled to a
broadband time-dependent squeezing vacuum reservoir}\label{example1}

Consider a $\Xi$-type atom coupled to a time-dependent broadband
squeezed vacuum field \cite{geo1,geo2}.  For squeezed vacuum field,
the initial state can be written as
\begin{eqnarray}
\ket{vac(\eta)}= K(\eta)\ket{vac}, \nonumber
\end{eqnarray}
where $K(\eta)$ is a multi-mode squeezing transformation. There are
three classical fields $\Omega_j(t) (j=1,2,3)$ interacting with the
atom as shown in FIG.\ref{three}. Under the Born-Markovian
approximation, the evolution of the atom is governed by the
following master equation \cite{geo1},
\begin{eqnarray}
\dot \rho(t)=&&-i[
H(t),\rho(t)]-\frac{\gamma}{2}\{R^\dag(\eta)R(\eta)\rho(t)\nonumber\\
&&+\rho(t)R^\dag(\eta)R(\eta)-2R(\eta)\rho(t)R^\dag(\eta)\},\label{maeq1}
\end{eqnarray}
where
$H(t)=[\Omega_1(t)\ket{1}\bra{0}+\Omega_2(t)
\ket{0}\bra{-1}+\Omega_3(t)\ket{1}\bra{-1}+h.c.]$,
and
\begin{eqnarray}
R(\eta)=S\cosh(r)+\exp(i\phi)S^\dag\sinh(r).\label{up1}
\end{eqnarray}
The operator $S=\ket{1}\bra{0}+\ket{0}\bra{-1}$ denotes the
absorption of an excitation from time-dependent broadband squeezed
vacuum field, and $\eta(t)=r\exp(i\phi)$ is time-dependent squeezing
parameter with polar coordinates $\phi\in[0,2\pi]$ and $r>0$  called
phase and amplitude of the squeezing, respectively. Form
Eq.(\ref{up1}), the state
\begin{eqnarray}
\ket{\Phi_{\text{DF}}(r,\phi)}=c(r)\ket{-1}-\exp(i\phi)s(r)\ket{1},
\end{eqnarray}
with $c(r)=\cosh(r)/\sqrt{\cosh(2r)}$ and
$s(r)=\sinh(r)/\sqrt{\cosh(2r)}$, satisfies
$R(\eta)\ket{\Phi_{\text{DF}}(t)}=0$. Here we assume that the phase
of squeezing parameter is time-dependent with a  linear trend
$\phi=\omega_0 t$, and the amplitude of squeezing parameter is
constant.

\begin{figure}
\includegraphics[scale=0.5]{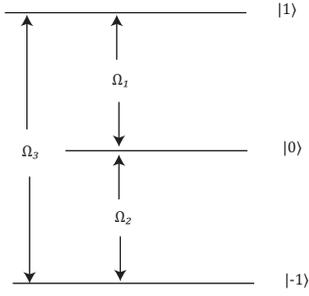}
\caption{(Color online) Schematic energy levels. A $\Xi$-type
three-level system driven by three classical fields $\Omega_j(t).$
We treat this system as an open system since we will consider it
coupled to a broadband squeezed vacuum reservoir.}\label{three}
\end{figure}

It has been shown in Refs.\cite{geo1,geo2} that the decoherence-free
evolution can be achieved by adiabatically changing the squeezing
phase. Comparing with Ref.\cite{geo1}, we will show here that the
adiabatic limit can be removed  when the time-dependent Hamiltonian
in Eq.(\ref{maeq1}) is manipulated  according to
Eq.(\ref{relation}). In other words, the open system could evolve
unitarily in the t-DFS by employing the scheme proposed in
Sec.\ref{SNC}. The master equation Eq.(\ref{maeq1}) can also be
realized in the model of a pair of trapped four-level
atoms\cite{geo3}, hence there are many manner to realize it in
experiment.

In the following, the state $\ket{\Phi_{\text{DF}}(t)}$ is chosen as
the basis of the t-DFS,  whose complemental space is then spanned by
$\ket{\Phi^\bot_1(t)}=s(r)\ket{-1}+\exp(i\omega_0 t)c(r)\ket{1}$ and
$\ket{\Phi^\bot_2(t)}=\ket{0}$. According to Eq.(\ref{relation}),
the classical fields $\Omega_j(t)$ must be modulated as follow,
\begin{eqnarray}
\Omega_1(t)&=&\cosh(r)\exp(i\omega_0t),\nonumber\\
\Omega_2(t)&=&\sinh(r),\nonumber\\
\Omega_3(t)&=&\omega_0\sinh(r)\cosh(r)\exp(i\omega_0t),\label{cf1}
\end{eqnarray}
such that $\ket{\Phi_{\text{DF}}(t)}$ forms a t-DFS.

The numerical results are presented in FIG.\ref{p1}, in which the
initial state is chosen to be
$\ket{\Phi_{\text{DF}}(0)}=c(r){\ket{-1}}-s(r)\ket{1}$ with $r=1$.
We plot  the purity $P(t)=\Tr[\rho^2(t)]$ with the parameter
$\omega_0=0.1\gamma$ (FIG.\ref{p1}(a)) and $\omega_0=10\gamma$
(FIG.\ref{p1}(b)) as a function of time. From FIG.\ref{p1}, we can
see that the control fields $\Omega_j(t)$ play an important role in
this scheme. Without the control fields (solid line in
FIG.\ref{p1}), the decoherence-free evolution can be approximately
realized only  under the adiabatic limit, i.e., the smaller  the
$\omega_0$ is, the better the state remains in the t-DFS. This can
be understood by examining  Eq.(\ref{relation}): When the control
fields are absent, the t-DFS condition is equivalent to that the
right side of Eq.(\ref{relation}) is  zero, i.e.,
$\omega_0\rightarrow0$. In other words, to make sure that the
evolution of quantum state is decoherence free, the squeezing phase
has to change adiabatically. As expected, when the open system are
controlled by the classical fields $\Omega_j(t)$ as Eq.(\ref{cf1})
(dash line in FIG.\ref{p1}), the evolution of the quantum state is
always unitary regardless of how fast the squeezing phase changes.
\begin{figure}
\includegraphics[scale=0.5]{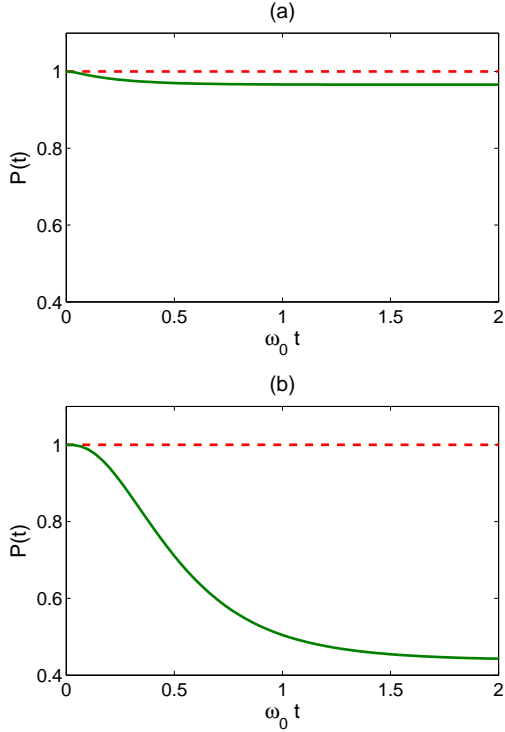} \caption{(Color online) The purity
$P(t)$ as a function of $\omega_0 t$ (in units of $\pi$) with
time-dependent Hamiltonian (dash line) and time-independent
Hamiltonian (solid line) for (a) $\omega_0=0.1\gamma$ and (b)
$\omega_0=10\gamma$. The the other parameters chosen are $\gamma=1$,
$r=1$.}\label{p1}
\end{figure}

\subsection{A toy model for DFS with time-dependent dimension}\label{example3}

In this subsection, we present a toy model with varying dimension of
t-DFS.
\begin{figure}
\includegraphics[scale=0.55]{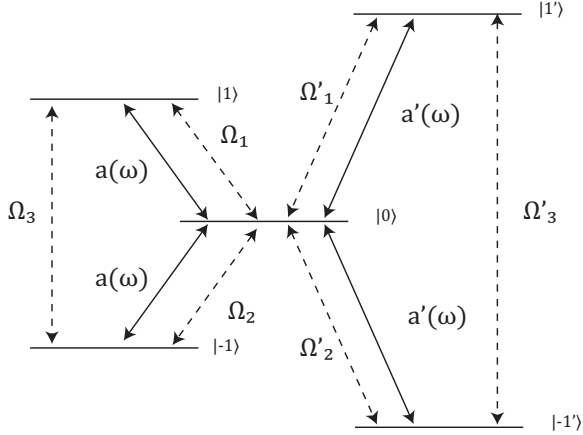}
\caption{(Color online) Schematic energy diagram. A five-level
system, transition $\ket{1}\leftrightarrow\ket{0}$ and
$\ket{0}\leftrightarrow\ket{-1}$ are driven by mode $a$,
$\ket{1'}\leftrightarrow\ket{0}$ and
$\ket{0}\leftrightarrow\ket{-1'}$ are driven by mode $a'$, and
$\ket{1'}\leftrightarrow\ket{-1'}$ by $b$.}\label{5l}
\end{figure}
Consider a five-level system coupled to two different broadband
squeezing vacuum fields and six classical control fields ($\Omega_j$
and $\Omega'_j$), as shown in FIG.\ref{5l}. $a(\omega)$ and
$a'(\omega$) are the annihilation operators of the environment with
different polarization and frequency $\omega$. Suppose the squeezing
vacuum states for the modes $a(\omega)$ and $a'(\omega)$ have
different squeezing parameters $\eta_1(t)=r_1\exp(i\omega_0 t)$ and
$\eta_2(t)=r_2\exp(i\omega_0 t)$, respectively. Under the same
assumptions used in Eq.(\ref{maeq1}), the master equation can be
written as \cite{geo1},
\begin{eqnarray}
\dot{\rho}(t)&=&-i[{H}(t),\rho(t)]\nonumber\\
&&+\sum_{\alpha=1}^2 \gamma_\alpha \left[F_\alpha(t) \rho(t)
F_\alpha^\dag(t)-\frac{1}{2}\{F_\alpha^\dag(t) F_\alpha(t),
\rho(t)\}\right],\nonumber\\
\end{eqnarray}%
where the Lindblad operators are
$F_\alpha(t)=\cosh(r_\alpha)S_\alpha+\exp(i\omega_0
t)\sinh(r_\alpha)S_\alpha^\dag$ ($\alpha=1,2$) with
$S_1=\ket{1}\bra{0}+\ket{0}\bra{-1}$ and
$S_2=\ket{1'}\bra{0}+\ket{0}\bra{-1'}$. The Hilbert space $\mathcal
H_S(t)$ of open system can be spanned by the following orthogonal
normalized bases,
\begin{eqnarray}
&&\ket{\Phi_{\text{DF1}}(t)}=c_1\ket{-1}-\exp(i\omega_0 t)s_1\ket{1},\\
&&\ket{\Phi_{\text{DF2}}(t)}=c_2\ket{-1'}-\exp(i\omega_0 t)s_2\ket{1'},\\
&&\ket{\Phi^\bot_1(t)}=s_1\ket{-1}+\exp(i\omega_0 t)c_1\ket{1},\\
&&\ket{\Phi^\bot_2(t)}=s_2\ket{-1'}+\exp(i\omega_0 t)c_2\ket{1'},\\
&&\ket{\Phi^\bot_0(t)}=\ket{0},
\end{eqnarray}
with $s_i{=}\sinh(r_i)/\sqrt{\cosh(2r_i)}$ and $c_i{=}\cosh(r_i)/
\sqrt{\cosh(2r_i)}$. It is not difficult to check that this model
admits a two-dimensional t-DFS at most,
\begin{eqnarray}
&&F_1(t)\ket{\Phi_{\text{DF1}}}=F_2(t)\ket{\Phi_{\text{DF1}}}=0,\\
&&F_1(t)\ket{\Phi_{\text{DF2}}}=F_2(t)\ket{\Phi_{\text{DF2}}}=0.
\end{eqnarray}
The Hamiltonian
\begin{eqnarray}
&& H(t)={H}_0(t)+T(t)\Omega(\ket{\Phi_{\text{DF1}}(t)}
\bra{\Phi_{\text{DF2}}(t)}+h.c.),\nonumber\\
\end{eqnarray}
includes two terms: The first term $ H_0(t)$ is required to
construct the t-DFS with time-dependent dimension,
\begin{eqnarray}
\text
H_0(t)&=&(\Omega_1(t)\ket{1}\bra{0}+\Omega_2(t)
\ket{0}\bra{-1}+\Omega_3(t)\ket{1}\bra{-1}\nonumber\\
&&+\Omega_4(t)\ket{1'}\bra{0}+\Omega_5(t)\ket{0}
\bra{-1'}+\Omega_6(t)\ket{1'}\bra{-1'}\nonumber\\&&+h.c.),
\end{eqnarray}
where the classical field strengthes, according to
Eq.(\ref{relation}), are designed as follow,
\begin{eqnarray}
\Omega_1(t)&=&\cosh(r_1)\exp(i\omega_0t),\nonumber\\
\Omega_2(t)&=&\sinh(r_1),\nonumber\\
\Omega_3(t)&=&\omega_0\sinh(r_1)\cosh(r_1)\exp(-i\omega_0t),\nonumber\\
\Omega'_1(t)&=&\cosh(r_2)\exp(i\omega_0t),\nonumber\\
\Omega'_2(t)&=&\left\{
                \begin{array}{ll}
                  0, & {\omega_0t<\pi/2;} \\
                  (1+\cos(2\omega_0t))\sinh(r_2), & {\pi/2\leq\omega_0t\leq\pi}; \\
                  \sinh(r_2), & { \omega_0t>\pi.}
                \end{array}
              \right.\nonumber\\
\Omega'_3(t)&=&\omega_0\sinh(r_2)\cosh(r_2)\exp(-i\omega_0t).\label{cf3}
\end{eqnarray}
\begin{figure}
\includegraphics[scale=0.55]{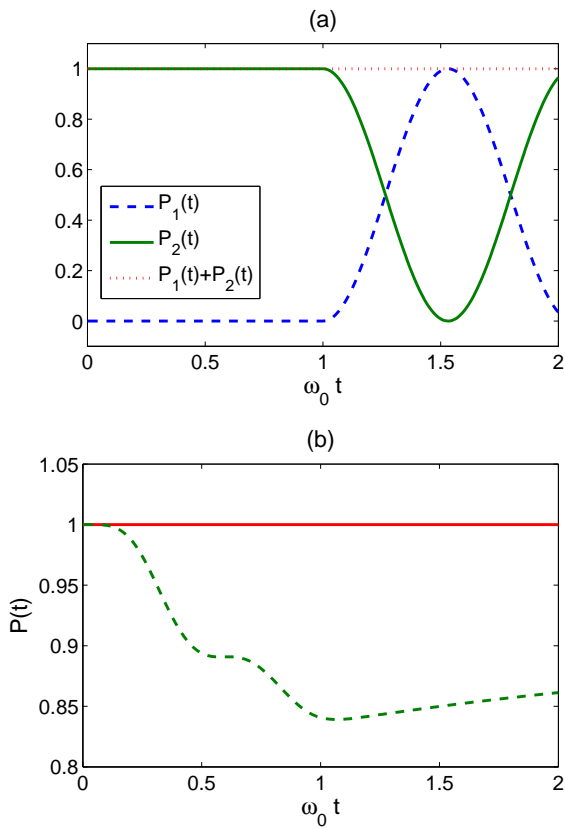}
\caption{(Color online) (a) The evolution of population on
$\ket{\Phi_{\text{DF1}}(t)}$ (solid line),
$\ket{\Phi_{\text{DF2}}(t)}$ (dash line) and population in $\mathcal
H_{\text{DFS}}(t)$ (dots line) with $r_1=r_2=1$ and
$\Omega=\omega_0=\gamma=1$. (b) The evolution of the purity $P(t)$
with $T(t)=1$ (dash line) and $T(t)=\theta(\omega_0t-\pi)$ (solid
line).} \label{p3}
\end{figure}
This implies that, for $0{<}\omega_0 t{\leq}\pi$, the t-DFS
$\mathcal H_{\text{DFS}}(t)$ is one-dimensional and spanned by
$\ket{\Phi_{\text{DF1}}(t)}$, since $\bra{\Phi^\bot_i(t)}\text
H_{\text{eff}}(t)\ket{\Phi_{\text{DF1}}(t)}=\bra{\Phi_{\text{DF2}}(t)}\text
H_{\text{eff}}(t)\ket{\Phi_{\text{DF1}}(t)}=0$. For $\omega_0
t>\pi$, the DFS $\mathcal H_{\text{DFS}}(t)$ is two-dimensional
spanned by $\ket{\Phi_{\text{DF1}}(t)},\,$ and $
\ket{\Phi_{\text{DF2}}(t)}$, since $\bra{\Phi^\bot_i(t)}\text
H_{\text{eff}}(t)\ket{\Phi_{\text{DF1}}(t)}=\bra{\Phi^\bot_i(t)}\text
H_{\text{eff}}(t)\ket{\Phi_{\text{DF2}}(t)}=0$. Thus,  the dimension
 of the t-DFS $\mathcal H_{\text{DFS}}(t)$ changes  with time.
To illustrate the dimension changing with time, we introduce the
second term of the Hamiltonian that induces a transition between
$\ket{\Phi_{\text{DF1}}(t)}$ and $\ket{\Phi_{\text{DF2}}(t)}$, where
$T(t)=\theta(\omega_0t-\pi)$ is a step function and $\Omega$ is
transition coefficient. This means that the transition is allowed
for $\omega_0 t>\pi$ but it is forbidden for $0<\omega_0 t\leq\pi$.

The numerical simulations are presented in FIG.\ref{p3}. The
population on $\ket{\Phi_{\text{DF1}}(t)}$ (dash line) and
$\ket{\Phi_{\text{DF2}}(t)}$(solid line),
\begin{eqnarray}
P_1(t)=\bra{\Phi_{\text{DF1}}(t)}\rho(t)\ket{\Phi_{\text{DF1}}(t)},\nonumber\\
P_2(t)=\bra{\Phi_{\text{DF2}}(t)}\rho(t)\ket{\Phi_{\text{DF2}}(t)},
\end{eqnarray}
are shown in FIG.\ref{p3}(a). As predicted, when $0<\omega_0
t\leq\pi$, the population stays on $\ket{\Phi_{\text{DF1}}(t)}$;
when $\omega_0 t>\pi$, the population transits between
$\ket{\Phi_{\text{DF1}}(t)}$ and $\ket{\Phi_{\text{DF2}}(t)}$, but
the total population $P_1(t)+P_2(t)$ (dot line in FIG.\ref{p3}(a))
remains unchanged, this means that the system does not leak out the
t-DFS. If the Hamiltonian does not satisfy the t-DFS condition, for
example, $T(t)=1$ instead of $T(t)=\theta(\omega_0t-\pi),$ the
purity $P(t)$ will decay (dash line in FIG.\ref{p3}(b)). This
implies that the dimension of the t-DFS really changes in the time
evolution at time $t_0$ given by $\omega_0 t_0 =\pi.$ Once the
dimension of the t-DFS changes, the time-independent DFS and t-DFS
can not be connected by an unitary transformation. In the other
words, the t-DFS is not trivial.

\section{Conclusion}\label{conclusion}
In summary, we have defined and presented a necessary and sufficient
condition for  the time-dependent decoherence-free subspace (t-DFS)
for open systems. In contrast to the time-independent DFS, the basis
of this t-DFS is time-dependent. Besides, the dimension of t-DFS may
change with time, this  implies that we can not trivially  get the
t-DFS by unitary transformations. Two examples are presented to
illustrate the t-DFS. In the first example, we show in details how
to manipulate the Hamiltonian of a $\Xi-$type system to realize a
one-dimensional t-DFS, while in the second example, we through a toy
model to show that the dimension of the t-DFS can change with time.
The later  indicates that the t-DFS can not be derived by an unitary
transformation  from the conventional DFS.

{The observation  of the t-DFS and the prediction made for the t-DFS
is within reach of recent technology. In fact, in a recent proposal
\cite{cirac3}, the authors proposed a scheme to entangle two atomic
ensemble of Cesium at room temperature  by  engineering reservoir
\cite{cirac1,cirac2}. These techniques together with measurement
\cite{cirac1}, can realize the t-DFS. For example, the
time-dependent Lindblad operators may be achieved by modulating the
detuning between the pumping field and the atom or by modifying the
Zeeman splitting. (A detail relation between Lindblad operators and
parameters of environment can be found in Refs. \cite{cirac3}.) The
parameters in the Lindblad operators can be modulated in experiment
\cite{detu}. All these together can lead to the  t-DFS. In practice,
the Lindblad operators are not known {\it a priori}. However, we can
re-design a decoherence-free subspace by controlling the Hamiltonian
and engineering the reservoir. This means that we can improve the
decoherence by designing a DFS to include the states of interests.}

We thank X. L. Huang for  valuable discussions on this manuscript.
This work is supported by the NSF of China under Grants  No
10935010, No 11175032 and No 10905007.

\end{document}